\documentclass[journal,onecolumn]{IEEEtran}

\usepackage{times,amsmath,amsthm,epsfig,amssymb,graphicx,amsfonts}
\usepackage{hyphenat}
\usepackage{psfrag}
\usepackage{ifpdf}
\usepackage{cite}
\usepackage{array}
\usepackage{url}

\usepackage{multirow}
\usepackage{graphicx}
\usepackage{color}
\usepackage{subfigure}
\usepackage{epsfig}
\usepackage{citesort}
\usepackage{setspace}
\usepackage{latexsym}
\usepackage{amssymb}
\usepackage{amsmath}

\newcommand{\m}{$\mathsf{M}$}

\title{Distributed Data Storage Systems with Opportunistic Repair}
\author{Vaneet Aggarwal, Chao Tian, Vinay A. Vaishampayan, and Yih-Farn R. Chen \thanks{V. Aggarwal, and Y. R. Chen are with  AT\&T Labs-Research, Bedminster, NJ 07921, email: \{vaneet, chen\}@research.att.com. C. Tian is with the Department of Electrical Engineering and Computer Science, University of Tennessee, Knoxville, TN, email: chao.tian@utk.edu. V. A. Vaishampayan is with the Department of Engineering Science and Physics, City University of New York, Staten Island, NY 10314, email: vavaishmapayan@icloud.net. This work was done in part when the authors were with AT\&T Labs-Research, Florham Park, NJ 07932.

This work was presented in part at IEEE Infocom 2014.}}
\newtheorem{theorem}{Theorem}

\newtheorem{lemma}{Lemma}
\newtheorem{corollary}{Corollary}

\begin{document}
\maketitle
\begin{abstract}
The reliability of erasure-coded distributed storage systems, as
measured by the mean time to data loss (MTTDL), depends
on the repair bandwidth of the code. Repair-efficient codes provide
reliability values several orders of magnitude better than conventional
erasure codes. Current state of the art codes fix the number
of helper nodes (nodes participating in repair) a priori. In practice, however, it is
desirable to allow the number of helper nodes to be adaptively
determined by the network traffic conditions. In this work, we propose
an opportunistic repair framework to address this issue. It is shown
that there exists a threshold on the storage overhead, below which
such an opportunistic approach does not lose any efficiency from the optimal storage-repair-bandwidth tradeoff; i.e. it is possible to construct a code simultaneously optimal for different numbers of helper nodes.
We further examine the benefits of such opportunistic codes, and
derive the MTTDL improvement for two repair models: one
with limited total repair bandwidth and the other
with limited individual-node repair bandwidth. In both settings,
we show orders of magnitude improvement in MTTDL. Finally, the proposed framework is examined in a network setting where a significant improvement in MTTDL is observed.
\end{abstract}

\section{Introduction}

Efficient data storage systems based on erasure codes have attracted much attention recently, because they are able to provide reliable data storage at a fraction of the cost of those based on simple data replication. In such systems, the data shares stored on a server or disk may be lost, either due to physical disk drive failures, or due to storage servers leaving the system in a dynamic setting. To guarantee high reliability, the lost  data shares must be repaired and placed on other servers. The total amount of data that needs to be transferred during this repair phase should be minimized, both to reduce network traffic costs and to reduce repair time. Dimakis {\em et al.}  \cite{Dimakis:10} recently proposed a framework to investigate this tradeoff between the amount of storage at each node ({\em i.e.,} data storage) and the amount of data transfer for repair ({\em i.e.,} repair bandwidth).

In the setting considered in \cite{Dimakis:10}, there are $n$ nodes, data can be recovered from any $k$ nodes, and the lost data share needs to be regenerated using (certain information obtained from) $d\ge k$ helper nodes. It was shown that for any fixed $d>k$, there exists a natural tradeoff between the total amount of repair traffic bandwidth and the storage overhead; the two extreme points are referred to as the minimum storage regenerating (MSR) point and the minimum bandwidth regenerating (MBR) point, respectively. In general, by using $d>k$ helper nodes, the total repair traffic can be reduced, compared to the naive (and currently common) practice of using only $k$ helper nodes.

There are two repair strategies considered in the literature. The first is functional repair in which the coded information on the repaired node may differ from that of the failed node, 
but the repaired system maintains the code properties.  The code construction for functional repair in \cite{Dimakis:10} uses network coding \cite{Yeung:00}, and a survey for using network coding in distributed storage can be found in \cite{Dimakis:11}. The second strategy, exact repair, requires that the failed disk is reconstructed exactly. Previous work for exact repair includes \cite{RashmiShah:12:1}\cite{RashmiShah:12:2}\cite{RashmiShah:11}\cite{Cadambe:11}\cite{PapailiopoulosDimakisCadambe:11}\cite{Tamo:11}\cite{CadambeHuang:11}\cite{Chao}\cite{evenodd}. In this paper, we will consider functional repair for distributed storage systems.

The number of helper nodes, $d$, is a design parameter in \cite{Dimakis:10}; however, in practice it may not be desirable to fix $d$ a priori. As an example, consider the dynamic peer-to-peer distributed storage environment, where peers are geo-distributed in a wide area without a centralized control mechanism. The peers may choose to join and leave the system in a much less controlled manner than in a data center setting. One example of such a system is the Space Monkey project \cite{spacemonkey}, and another is the open source peer-to-peer storage sharing solution built in Tahoe-LAFS \cite{tahoe}. In such systems, at the time of node repair it is desirable to utilize as many helper nodes as possible to achieve the most efficient repair, instead of accessing only a fixed subset of $d$ available nodes. In other words, instead of designing a code for a single value of $d$, it would be desirable for a code to be universally applicable for multiple values of $d$.


In this work, we address this problem and propose an erasure coding approach with opportunistic repair, where a single universal code is able to regenerate the share on the lost node by using any $d$ helper nodes, in a set $D$ of available choices of $d$ values. We develop fundamental bounds on the performance, and evaluate the performance improvement in several example scenarios.

We first investigate the tradeoff between storage and the repair bandwidths (one bandwidth for each $d$ value), and provide a characterization of the complete tradeoff region through an analysis technique similar to that used in \cite{Dimakis:10}. A particularly interesting question arises as to whether there is a loss for any given $d\in D$ by taking this opportunistic approach, when compared to the case in \cite{Dimakis:10} where the parameter $d$ is fixed a-priori. We find that there is a critical storage overhead threshold, below which there is no such loss; above this threshold, the universality requirement of the code indeed incurs a loss. In particular, the MSR points for individual $d$ values are simultaneously achievable in the opportunistic setting; this phenomenon for the special case of MSR codes has  in fact been observed in an asymptotically optimal  code construction \cite{Cadambe:11}.

The reliability of storage systems is usually a foremost
concern to the service provider and the users. Data loss events can
be extremely costly: consider for example the value of
data in systems storing financial or medical records. For
this reason, storage systems must be engineered such
that the chance of an irrecoverable data loss is extremely
low, perhaps on the order of one in 10 million objects per
10,000 years of operation - as is the case of the Amazon Glacier Storage Service.
The reliability of distributed storage systems is usually measured using the mean time to data loss (MTTDL) of the system. We analyze the reliability under two models, the first in which different failed disks are repaired serially one after the other, and the second with parallel repair. For both of these models, we show that the MTTDL  is improved by a multiplicative factor of $(n-k)!$, when compared to that without opportunistic repair. This translates to a significant improvement, and is usually many orders of magnitude better, even for lower values of $n$ and $k$. For $n=51$ and $k=30$, this improvement is a multiplicative factor of around $5.1 \times 10^{19}$, while even for the parameters chosen in Facebook HDFS ($n=14$ and $k=10$ ) \cite{facebook}, this improvement is a factor of 24.


In practical systems, the bandwidths from all the nodes are not the same. Hence, there is an additional optimization step: choose a subset of nodes among the active ones to repair a failed node. Even with different bandwidths of different links among the nodes, we find that the opportunistic distributed storage system helps in repair time and hence the MTTDL increases by orders of magnitude.

The remainder of the paper is organized as follows. Section II gives the background and introduces  the system model.   Section III gives our results on functional repair, and the loss for an opportunistic system as compared to an optimized system for a single value of $d$. Section IV analyzes the mean time to data loss of an opportunistic distributed storage system. Section V provides performance comparisons in a few example cases, and Section VI concludes the paper.

\section{Background and System Model}
\subsection{Distributed Storage Systems}

\begin{figure}[tb] 
   \centering
   \includegraphics[trim=0in 3in .3in 2in, clip, width=0.8\columnwidth]{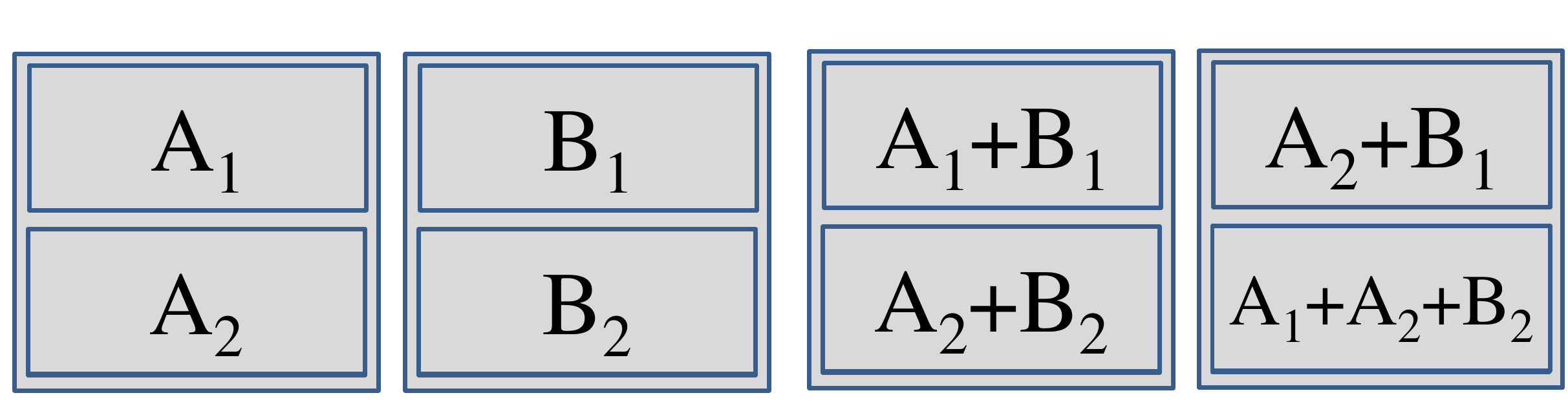}
   \caption{A (4,2) MDS binary erasure code (\cite{evenodd}).
Each storage node stores two blocks that are linear binary
combinations of the original data blocks $A_1$, $A_2$, $B_1$ and $B_2$, where the total file size is $\mathsf{M} = 4$ blocks. }
   \label{fig:code_place}
\end{figure}

A distributed storage system consists of multiple storage nodes that are  connected in a
network. The issue of code repair arises when a storage node of the system fails. Consider an example with $n=4$ distributed nodes, which can be recovered from any $k=2$ nodes as shown in Figure \ref{fig:code_place}. If only one node fails and needs to be repaired, the conventional strategy is to get all the $M=4$ blocks from two active nodes. However, as shown in \cite{Dimakis:10}, fewer than $M$ blocks are sufficient for repairing a failed node and in fact three blocks suffice in the above example. For example, to repair the first node, $B_1$, $A_1+B_1$, and $A_2+B_1$ are enough to get $A_1$ and $A_2$. Similarly, to repair the second node, $A_2$, $A_2+B_2$ and $A_2+B_1$ are sufficient to get $B_1$ and $B_2$. To repair the third node, $A_1-A_2$, $-B_1+B_2$ and $A_2+B_1$ are sufficient to get $A_1+B_1$ and $A_2+B_2$. To repair the fourth node, $A_1$, $B_1-B_2$ and $A_2+B_2$ are sufficient to get $A_2+B_1$ and $A_1+A_2+B_2$.

In this paper, we will consider a mode of repair called functional repair, in which the failed block may be repaired with possibly different data than that of the failed node, as long as the repaired system maintains the code properties (repair bandwidth and the ability to recover from $n-k$ erasures). The authors of \cite{Dimakis:10} derive a tradeoff between the amount of storage at each node and the repair bandwidth for a given number of nodes $n$, number of nodes from which the data should be recovered $k$ and the number of nodes $d$ that can be accessed for repair. Let each node store $\alpha$ bits, let $\beta_d$ bits be downloaded from each of the $d$ nodes for repair and let the total data be \m\,  bits. It was shown in \cite{Dimakis:10} that the optimal storage-repair-bandwidth tradeoff, i.e., $\alpha$ vs. $\beta_d$, satisfies
\begin{equation}
 \sum_{i=0}^{k-1}\min(\alpha, (d -i)\beta_{d}) \ge \mathsf{M}. \label{betastar}
\end{equation}
For a given $\alpha$, the minimum $\beta_d$ satisfying the above is given as $\beta_d^*(\alpha)$. The codes associated with  the two extremes of this tradeoff are referred to as minimum-storage regenerating (MSR)  and minimum-bandwidth regenerating (MBR) codes. An MSR code has a minimum storage overhead requirement per node while an MBR point has the minimum repair bandwidth.

One particularly important observation in this example code is as follows: there are two possible choices of the value $d$, which is $d=3$ (with repair bandwidth from each node as 1) and $d=2$ (with repair bandwidth from each node as 2). It is natural to ask whether this flexible choice of $d$ is generally available, and whether this flexibility incurs any loss of efficiency.

\subsection{Motivating Example in a Simple Network}
\begin{figure}[tb] 
   \centering
   \includegraphics[trim=.8in 1in 2.2in 1.3in,clip,width=1.5in]{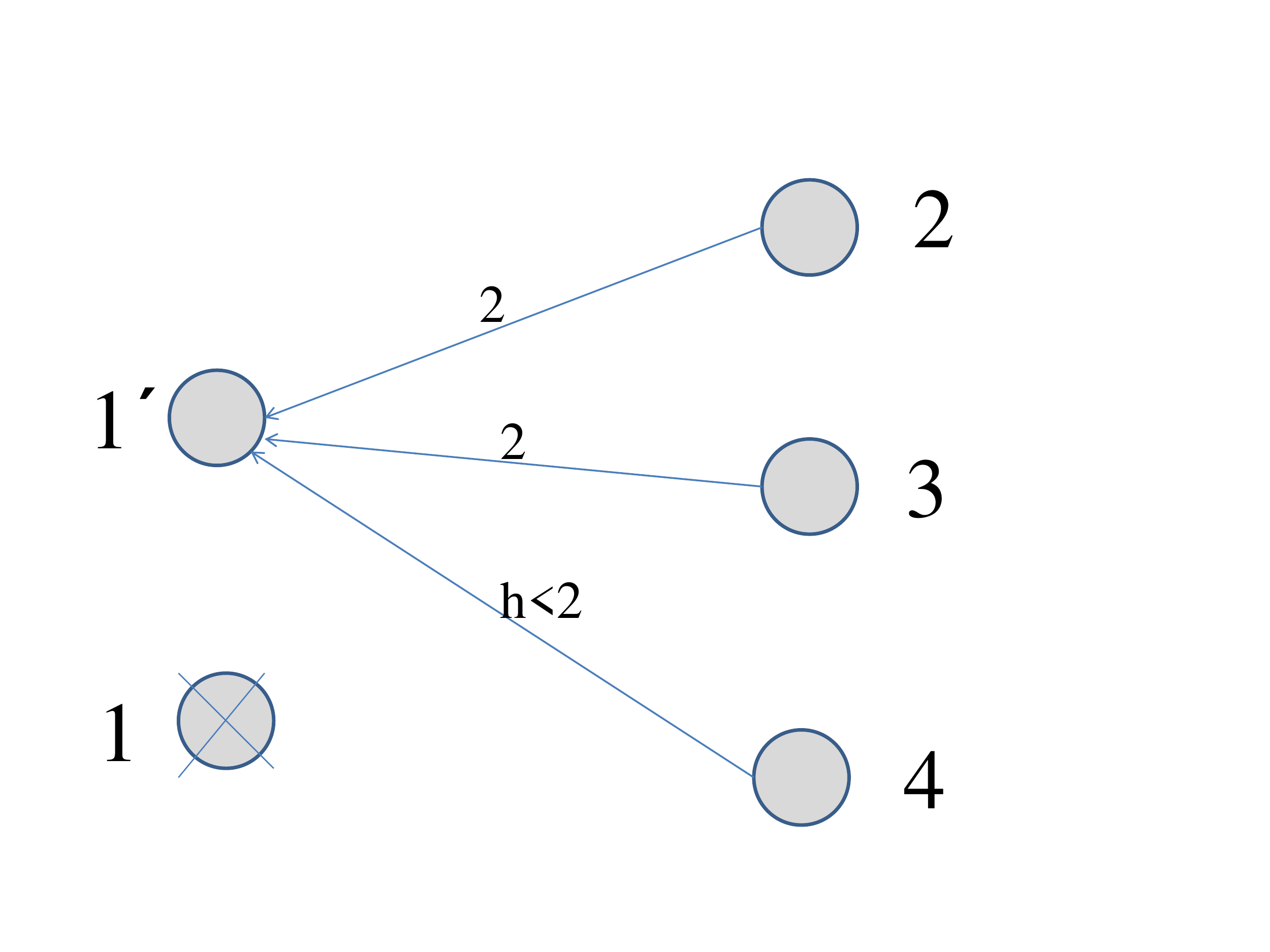}
   \caption{Example network with less bandwidth from one node.}
   \label{fig:ntwk}
\end{figure}

Consider  a file of size {\m} is encoded into $n$ shares, which are placed on $n$ distinct nodes\footnote{We do not distinguish between node and disk from here on.}.  We assume that the $n$ nodes are connected via a network and that $R_{ij},~1 \leq i,j \leq n,~i \neq j$ represents the bandwidth from node $i$ to node $j$. In a typical scenario the rates will obey a set of constraints and these constraints will have an impact on the time to repair a single failed node. As an illustration consider the network in Fig.~\ref{fig:ntwk}, which shows a failed node (node 1), an $(n,k)=(4,2)$ MDS code (such as a Reed-Solomon code), and the link bandwidths between the helper nodes and the node to be restored (node $1'$). Consider two scenarios, the first where $d=k=2$, the second where $d=n-1=3$. In the first case the repair time is proportional to $\mathsf{M}/4$ if the helper nodes are nodes 2 and 3. On the other hand if $d=3$, then even though we need to download only $\mathsf{M}/4$ symbols from each helper node, the bottleneck link is the link from the fourth node and thus the repair time is proportional to $\mathsf{M}/(4h)$. If $h>1$, it is beneficial to use $d=3$ while if $h<1$, it is beneficial to use $d=2$. This example illustrates the situations where the network topology may be unknown at encoding time, which may be changing due to link failures, upgrades, or, a network composed of mobile storage nodes.  Clearly there are benefits if  the value of $d$ and the set of helper nodes can be chosen at the time of repair.

Motivated by this observation, we assume a distributed storage system with probing, by which the number of nodes to access for repair can be determined before the repair process starts. Using an active probe, we obtain the bandwidths between any pair of nodes which can be used to decide the number of nodes to access in order to repair a failed node. Thus, in practice, the value of $d$ is not fixed and hence the code design should work for multiple values of parameter $d$.  We will consider bandwidths between multiple nodes for choosing the value of $d$ in Section V.

\subsection{Opportunistic Repair in Distributed Storage}

Assume that there are $n$ nodes and the total file that we need to store is of size \m, which can be reconstructed from $k$ nodes. The repair bandwidth is a function of how many nodes are used for repair. In this setting, we consider what happens for the region formed by the bandwidth to repair from multiple valued of nodes accessed, $d$. The code structure is assumed to remain the same and should allow for repair with varying values of $d\in D = \{d_1, d_2, \cdots \}$ such that each $d_i\ge k$. We denote the storage capacity at each node by $\alpha$ and the bandwidth from each node as $\beta_d$ when the content is accessed from $d$ nodes.

Thus, the distributed storage system with opportunistic repair works for different values of $d$. An opportunistic code can take advantage of varying number of failed nodes. For example, in the $n=4$, $k=2$ case in Figure \ref{fig:code_place}, we note that the system can be repaired from $d=2$ as well as $d=3$ nodes. So, this code design works for all $k \le d \le n-1$. In this paper, we will find the parameters the codes have to satisfy so that they work for multiple values of $d$. It was noted in \cite{Cadambe:11} that at the MSR point, there is no loss asymptotically for opportunistic repair in the sense that at $\alpha_{MSR}$, the optimal values $\beta_d$ for the MSR point corresponding to $d$ are simultaneously achievable for all $d\ge k$. This result was further extended in \cite{adaptive} for adaptive repair at the MSR point where the failed nodes perform a coordinated repair. In this paper, however, we do not consider any coordination among the failed nodes. Furthermore, we will consider the complete tradeoff region formed by the storage capacity and the repair bandwidths for different values of $d$.


\section{Results on Opportunistic Distributed Storage System}

In this section, we present the main theoretical results on opportunistic repair distributed storage systems.

As in storage systems with a fixed number of repair helper nodes, in opportunistic-repair distributed storage systems, there is also a fundamental tradeoff between the share size and the repair bandwidths. For functional repair, this tradeoff can be completely characterized as given in the following theorem.

\begin{theorem}
\label{thoerem:main}
If the value of $\alpha$ and $\beta_{d_j}$ for all $d_j\in D$ satisfies
\begin{equation}
\sum_{i=0}^{k-1}\min(\alpha, \min_{d_j \in D}(d_j -i)\beta_{d_j}) \ge \mathsf{M}, \label{main}
\end{equation}
there exist linear network codes that achieve opportunistic distributed storage system with the storage per node given by $\alpha$ and the repair bandwidth from $d_j$ nodes given by $d_j \beta_{d_j}$ for any $d_j \in D$. Further, if the above condition is not satisfied, it is information theoretically impossible to
achieve opportunistic distributed storage system with the above properties.
\end{theorem}
\begin{proof}
The proof follows by an extension of the result in \cite{Dimakis:10}. The details can be found in Appendix \ref{apdx_inner}.
\end{proof}

An important question of practical interest is whether by taking the opportunistic approach, a loss of storage-repair efficiency is necessarily incurred.  By leveraging Theorem \ref{thoerem:main}, we can provide an answer to this question, as given in the next theorem where  $D=\{d_1, \cdots d_l\}$ for $k\le d_l < \cdots < d_1 <n$. 

\begin{theorem}
\label{theorem:noloss}
For a given $n$, $k$, $\mathsf{M}$, and $\alpha\ge \frac{\mathsf{M}}{k}$, $(\alpha, \beta_{d_1}^*(\alpha),  \cdots, \beta_{d_l}^*(\alpha))$ satisfy \eqref{main} for $|D|>1$ if and only if either $k=1$ or $\alpha \le \alpha_o(k,d_1,M)\triangleq\frac{\mathsf{M}(d_1-k+2)}{k(d_1-k+2)-1}$.
\end{theorem}

The proof of this theorem is given in the appendix. This theorem essentially states that below the critical threshold $\alpha_o$ of the storage share size, there is no loss of optimality by imposing the opportunistic repair requirement, while above this threshold, such a loss is indeed necessary. A special case of practical relevance is given as a corollary next, which essentially states that there is no loss by requiring the MSR codes to have the opportunistic repair property.



\begin{corollary} MSR points for all values of $d$ are simultaneously achievable, i.e., $(\frac{\mathsf{M}}{k}, \beta_k^*(\frac{\mathsf{M}}{k}), \beta_{k+1}^*(\frac{\mathsf{M}}{k}), \cdots, \beta_{n-1}^*(\frac{\mathsf{M}}{k}))$ satisfy \eqref{main}.
\end{corollary}

This particular result has been previously observed in \cite{Cadambe:11}. In fact, even for the more stringent exact-repair case where the failed disk needs to be repaired with exact same copy, the same result holds asymptotically, using the class of asymptotically optimal codes constructed in \cite{Cadambe:11}.
%


The next theorem deals with the case when $(\alpha,\beta_{d_1})$ is operating on the optimal (non-opportunistic-repair) storage-repair-bandwidth tradeoff curve, when $\alpha$ is larger than the  given threshold $\alpha_o$. In this case, a loss of repair bandwidth is necessary for all the other values of $d=d_2,d_3,\ldots,d_l$, and the following theorem characterizes this loss precisely.

\begin{theorem}
For a given $n$, $k$, $\mathsf{M}$, and $\alpha\ge \frac{\mathsf{M}}{k}$, $(\alpha, \beta_{d_1}^*(\alpha), \frac{d_1 -k+1}{d_2-k+1} \beta_{d_1}^*(\alpha) \cdots, \frac{d_1 -k+1}{d_l-k+1} \beta_{d_1}^*(\alpha))$ satisfies \eqref{main}. Further, given that  $\beta_{d_1}=\beta_{d_1}^*(\alpha)$ in \eqref{main}, the above point has the smallest value of possible $\beta$'s for the remaining values $d_i$, $i>1$.
\end{theorem}
\begin{figure}[tb]
  \centering
      \includegraphics[trim = 1.15in 3in 1.5in 3.2in, clip, width=0.45\textwidth]{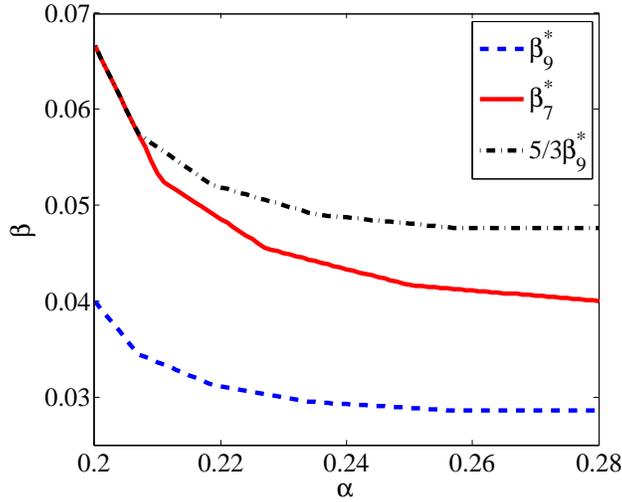}
  \caption{Loss with opportunistic repair for $n=10$, $k=5$, and $\mathsf{M}=1$. We consider the set of possible repair values as $\{7,9\}$.}\label{opp_example}
\end{figure}

As an example, we consider  $n=10$ and $k=5$ in Figure \ref{opp_example}. Note that without opportunistic repair, the tradeoffs for $d=7$ and $d=9$ can be independently achieved. However, both these curves are not simultaneously achievable. Till the first linear segment of the curve with largest $d$ ($d=9$ in this case), the values of $\beta$ on the two tradeoff curves are simultaneously achievable. After that, there is a loss. Assuming that we choose to be on the tradeoff for $d=9$, the black dash-dotted curve in Figure \ref{opp_example} represents the best possible tradeoff for $d=7$ and thus shows an increase with respect to the optimal code that works for only $d=7$.

\section{Mean Time to Data Loss}

In this section, we will consider the improvement in mean time to data loss (MTTDL)  due to opportunistic repair. There exist two models widely studied in the literature, which are used to evaluate the impact of opportunistic repair on these systems. We shall first provide necessary background on the models, and then discuss how these models fit in our proposed approach, and evaluate the performances.

\subsection{Chen\rq{}s Model and  Angus' model}

Both Chen\rq{}s model and Angus\rq{} model address systems with a total of $n$ components  (e.g., hard drives or nodes), which can withstand any $n-k$ component failures. The components may fail and are subsequently repaired at a certain rate, which can be modeled by a Markov Chain. The difference between the two models are in the rate of repairs.

The rate at which individual components fail is usually denoted by $\lambda$ while the rate at which those components are repaired is $\mu$. Alternatively, these
two rates might instead be expressed as times: Mean-Time-To-Failure (MTTF) and Mean-Time-To-Repair (MTTR) respectively. When $\lambda$ and $\mu$ are constant over
time, i.e.  exponentially distributed, $MTTF = 1/\lambda$
and $MTTR = 1/\mu$.

\begin{figure}[tb]
  \centering
      \includegraphics[trim = .6in 3.8in 2.8in .9in, clip, width=0.5\textwidth]{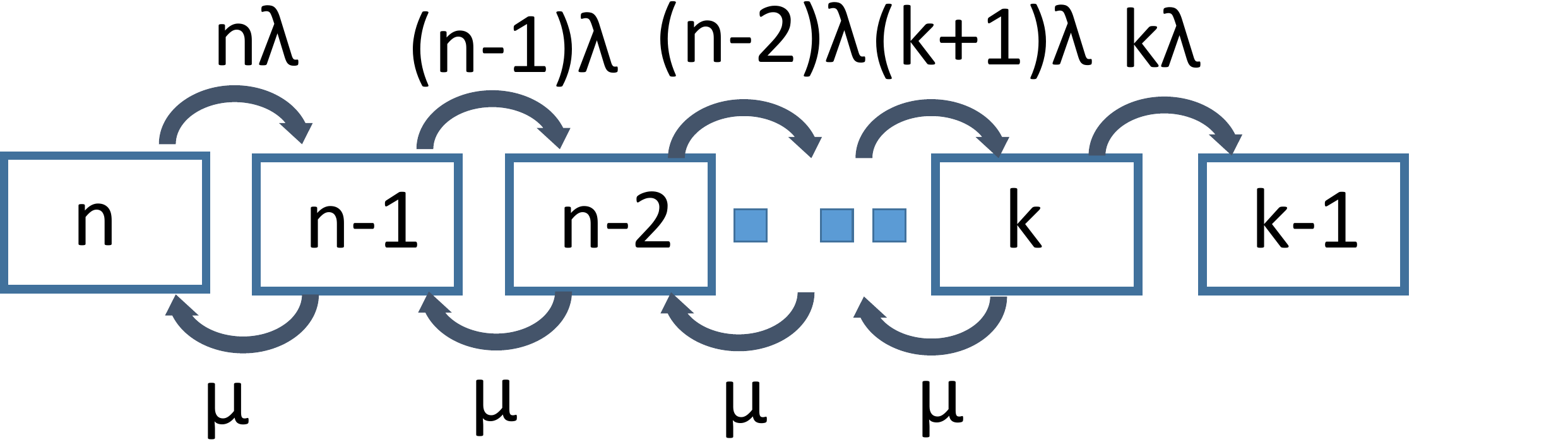}
  \caption{The state transition diagram for Chen's model.}\label{chen_org}
\end{figure}

\begin{figure}[tb]
  \centering
      \includegraphics[trim = .6in 3.8in 2.8in .9in, clip, width=0.5\textwidth]{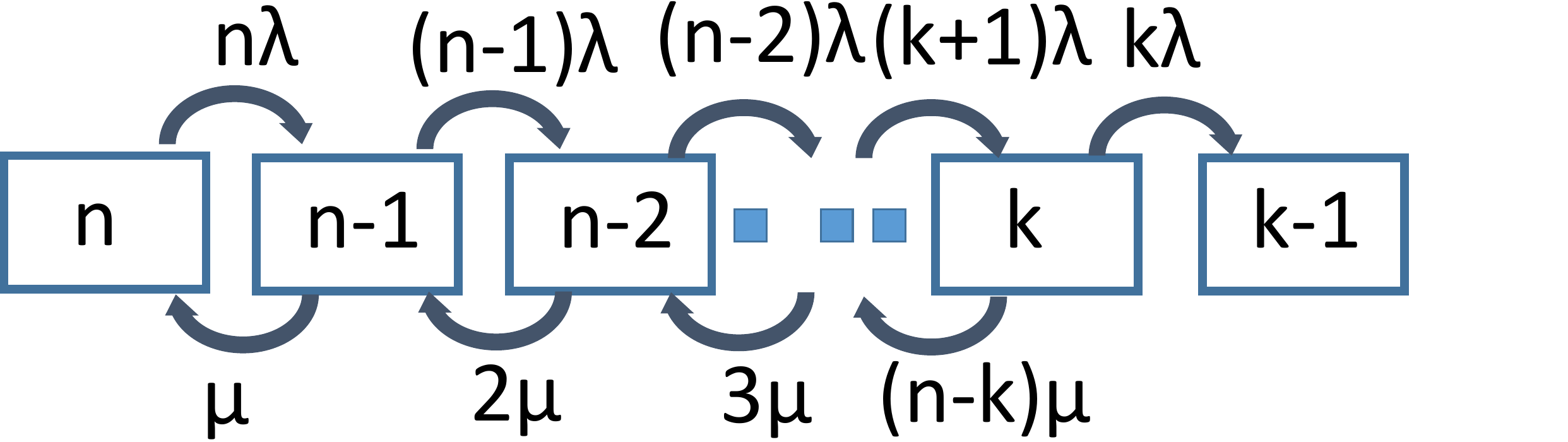}
  \caption{The state transition diagram for Angus' model.}\label{angus_org}
\end{figure}

\begin{figure}[tb]
  \centering
      \includegraphics[trim = .4in 2.7in 2.8in .9in, clip, width=0.5\textwidth]{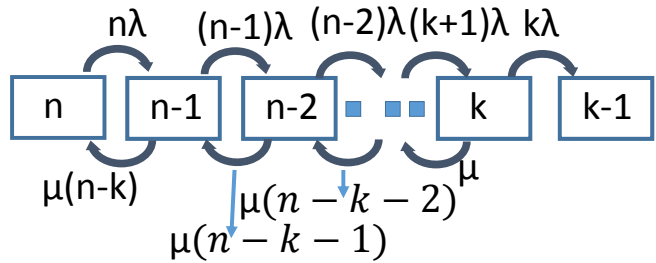}
  \caption{The state transition diagram for Chen's model with opportunistic repair.}\label{chen_mod}
\end{figure}

The first model is Chen's model \cite{Chen}. Chen et al. presented models for estimating the MTTDL for various RAID configurations, including RAID 0 (no parity), RAID 5 (single parity) and RAID 6 (dual parity). In this model,  failures occur at a rate equal to the number of operational
devices times the device failure rate, and repairs occur at the device
repair rate regardless of the number of failed devices. This model  for instance holds when one failed device is repaired at a time even when multiple devices fail. Thus, this model can also be called ``serial repair" where different failed devices are repaired one after the other. The repair and failure rates for the Chen's model are shown in Figure \ref{chen_org}.

The second well-accepted model in the literature is Angus' model \cite{Angus}. Unlike Chen's model, Angus' model assumes that there are unlimited repairmen. This means that whether
1 device or 100 fail simultaneously, each failed device
will be repaired at a constant rate. Thus, this model can be called  a ``parallel repair" model, where all the failed devices can be repaired simultaneously.  The state transition diagram for Angus' model is described in Figure \ref{angus_org}.

\subsection{MTTDL with Opportunistic Repair}

The \lq\lq{}repairman\rq\rq{} in Chen\rq{}s and Angus\rq{} models essentially captures the bottlenecks in the repair process, which can be either communication resources or computation resources. In the proposed framework, the bottleneck is mainly in the former. More precisely, the mean time to repair is inversely proportional to the bandwidth of the links that are used to repair the failed node. If all the incoming bandwidth of the communications links increase by factor $c>1$, mean time to repair goes down by factor $c$. For the same available communication links, if the repair traffic is reduced, the mean time to repair is reduced accordingly.

When opportunistic repair at the MSR point is used, the repair time is reduced if there are more than $k$ nodes in the system. Let us first consider Chen\rq{}s model, and using the result in the previous section, the parameters for state transition are as in Figure \ref{chen_mod}. Here the repair time is inversely proportional to the bandwidth needed from each of the remaining active device to repair the failed device. As an example, consider the transition from $n-1$ active nodes to $n$ active nodes. In this case, the repair bandwidth from each of the remaining $n-1$ nodes is a factor $n-k$ smaller than the repair bandwidth from each node when only $k$ nodes are used for repair. Thus, the effective data needed from a disk reduces by a factor $n-k$ and thus can be transferred in a factor of $n-k$ less time. Thus, the mean time to repair decreases by a factor $n-k$ thus giving the transition rate of $\mu(n-k)$.

Next, we will characterize the MTTDL for the Chen's model with opportunistic repair.

\begin{theorem} The MTTDL for Chen's model with opportunistic repair is given as
\begin{eqnarray}
MTDL_{\text{Chen, Opp}} &=& \sum_{l=0}^{n-k}\frac{(n-k-l)!}{(n-l)!}\sum_{i=0}^{n-k-l}\mu^i \nonumber\\
&&\lambda^{-(i+1)}\frac{(n-l-i-1)!}{(n-l-k-i)!}
\end{eqnarray}\label{thm_MTTDL}
\end{theorem}
\begin{proof}
The proof is provided in the Appendix \ref{apdx_MTTDL}.
\end{proof}

We note by the similar proof steps, the MTTDL for the original Chen's model is given as follows.
\begin{equation}
MTDL_{\text{Chen, Orig}} = \sum_{l=0}^{n-k}\frac{1}{(n-l)! }\sum_{i=0}^{n-k-l}\mu^i \lambda^{-(i+1)}(n-l-i-1)!
\end{equation}

We will now consider these expressions in the limit that $\lambda << \mu$. In this regime, the two expressions above are given as follows
\begin{eqnarray}
MTDL_{\text{Chen, Opp}} &=& \frac{(n-k)!(k-1)!}{n!}\frac{\mu^{n-k}}{\lambda^{n-k+1}}\\
MTDL_{\text{Chen, Orig}} &=& \frac{(k-1)!}{n!}\frac{\mu^{n-k}}{\lambda^{n-k+1}}
\end{eqnarray}

Thus, note that MTTDL increases by a factor of $(n-k)!$ with opportunistic repair as compared to that without opportunistic repair.

\begin{figure}[tb]
  \centering
      \includegraphics[trim = .4in 2.7in 2.8in .9in, clip, width=0.5\textwidth]{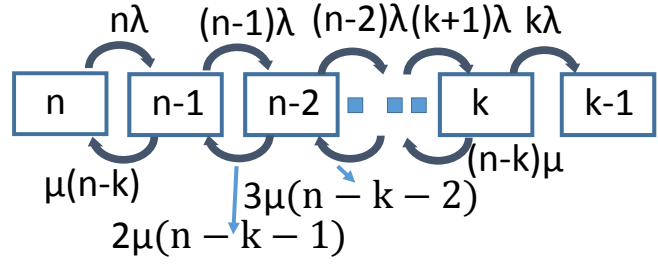}
  \caption{The state transition diagram for Angus' model with Opportunistic Repair.}\label{angus_mod}
\end{figure}

We can also use opportunistic repair for Angus' model and hence save bandwidth when there are more surviving nodes. Using opportunistic repair, the modified state transition is described in Figure \ref{angus_mod}. The MTTDL for Angus' model with opportunistic repair is given as follows.

\begin{theorem}  The MTTDL for Angus' model with opportunistic repair is given as
\begin{eqnarray}
MTDL_{\text{Angus, Opp}} &=& \sum_{l=0}^{n-k}\frac{(n-k-l)!}{(n-l)! }\sum_{i=0}^{n-k-l}\mu^i \nonumber\\
&& \lambda^{-(i+1)}\frac{(n-l-i-1)!}{(n-l-k-i)!}i!
\end{eqnarray}
\end{theorem}
Since the proof steps are similar to that in the Chen's model, the proof is omitted. Further,  the MTTDL for the original Angus' model is given as follows.
\begin{eqnarray}
&&MTDL_{\text{Angus, Orig}}\nonumber\\ &=& \sum_{l=0}^{n-k}\frac{1}{(n-l)! }\sum_{i=0}^{n-k-l}\mu^i \lambda^{-(i+1)}(n-l-i-1)! i!
\end{eqnarray}

We will now consider these expressions in the limit that $\lambda << \mu$. In this regime, the two expressions above are given as follows
\begin{eqnarray}
MTDL_{\text{Angus, Opp}} &=&  \frac{(k-1)!}{n!}\frac{\mu^{n-k}}{\lambda^{n-k+1}} ((n-k)!)^2 \\
MTDL_{\text{Angus, Orig}} &=& \frac{(k-1)!}{n!}\frac{\mu^{n-k}}{\lambda^{n-k+1}} (n-k)!
\end{eqnarray}

Thus,  the MTTDL increases by a factor of $(n-k)!$ with opportunistic repair as compared to that without opportunistic repair. Also,  note that the MTTDL loss for Angus' model is $(n-k)!$ higher than that in Chen's model.

\section{Network Simulation Examples}

We consider a network with $n$ distributed nodes, an $(n,k)$ MDS systematic code and consider the case where repair can be performed from any $d$ nodes, $k\le d<n$ nodes. When any node fails, it pings all the other nodes to determine  the bandwidths from each of the active node and then determines the helper set, i.e the nodes that will participate in the repair. Two scenarios are discussed next, the first one of which has deterministic bandwidth between nodes, while the second has random bandwidth between nodes.

\subsection{Deterministic scenario for distributed data centers}

\begin{figure}[htb] 
   \centering
   \includegraphics[trim =.1in 2.7in .1in 1.7in, clip,width=3.3in]{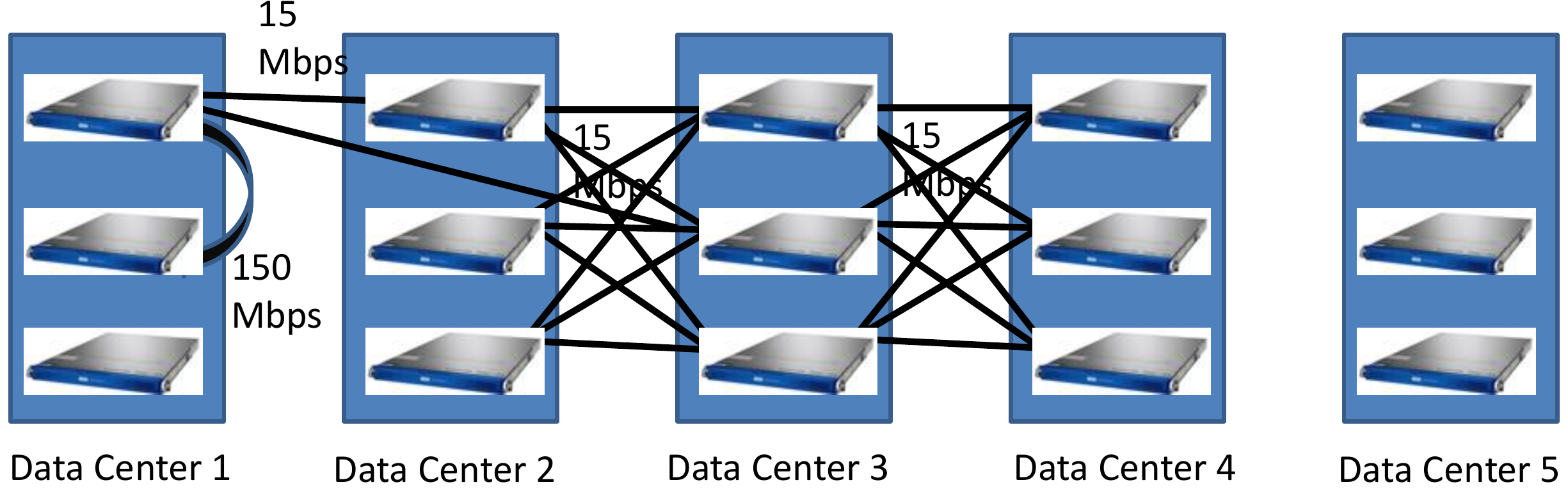}
   \caption{Five data centers, each with three nodes.}
   \label{fig:place_det}
\end{figure}

We consider a scenario where there are five distributed data centers. A (10,15) code is used and  three out of fifteen chunks are placed evenly on the five data centers. Let the bandwidth between the different storage nodes in the same data center be 150 Mbps while that in different data centers be  15 Mbps, as illustrated in Fig. \ref{fig:place_det}. We also consider a serial repair model.   When one node fails,  opportunistic repair chooses $d=14$ and the repair time for a $100$ Mb file is $\beta_{14}$/15 Mbps $ =2/15$ seconds. However, choosing $d=10$ gives the corresponding repair time of $\beta_{10}/15$ = $10/15$ seconds. Thus, there is a factor of five improvement in the repair time from a single failure with $d=14$ as compared to $d=10$. Similarly, there is a factor of $4$ improvement in repair rate with  2 failures ($d=13$ vs $d=10$), factor $3$ improvement in repair rate with  three failures and factor of $2$ improvement with four failures. If the repair process were exponentially distributed with these rates, the above analysis shows a performance improvement of a factor of $120$ as compared to the base system.  Thus,  the performance improvement at any stage of repair helps the system while opportunistic repair that helps use less bandwidth for multiple failure levels significantly outperforms choosing a single value of $d$.

\subsection{Random bandwidth links between different nodes}

\begin{figure}[htb] 
   \centering
   \includegraphics[trim =.1in 3.4in .1in 2.5in, clip,width=3.3in]{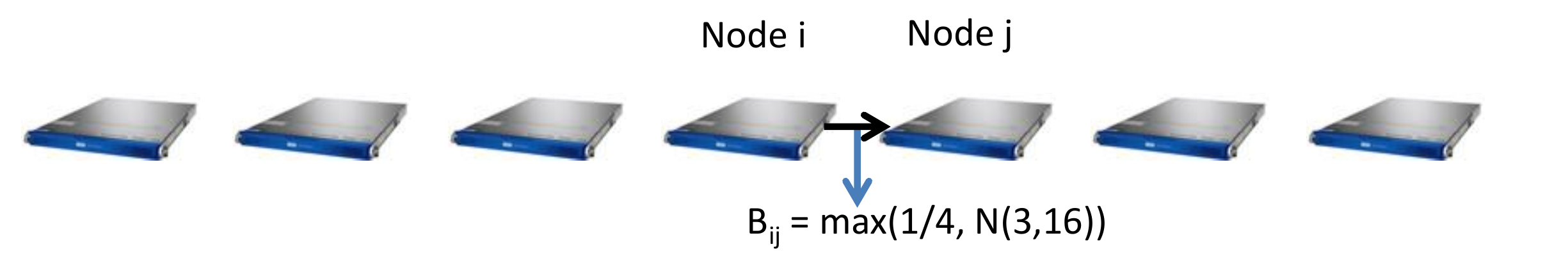}
   \caption{$n$ nodes with random bandwidth links between them}
   \label{fig:place_rand}
\end{figure}

We assume that there are $n$ nodes with the bandwidth between any two nodes is given by a maximum of a Gaussian random variable with mean 3 and variance 16 and 1/4, as illustrated in Figure \ref{fig:place_rand}. The bandwidth between any pair of nodes is independent. If we use  $d$ nodes ($d>k$) to repair (among the active nodes), there is an advantage in repair time based on Chen's model as $(d-k+1)$ which makes it beneficial to use as many nodes as possible to repair. However, in the case of realistic bandwidth, this rate is multiplied by $d^{\text{th}}$ best bandwidth to the node which is to be repaired since the data needs to be downloaded from $d$ servers. This rate decreases with $d$. Overall repair rate as a function of $d$ is the product of the two and is given as $(d-k+1) \mu $ times the $d^\text{th}$ largest bandwidth. Thus, there is an optimization needed for $d$ in order to select the value of $d$ to use. Since only a single node is repaired at a time, we choose the node with the maximum rate to repair.

We note that in practice, the value of $\mu$ may not be constant and change with time, and that multiple failed disks can be repaired in parallel. However, in this section, we ignore these factors, and only consider the value of $\mu$ to be constant and that one disk is repaired at a time.


We first see the mean time to repair when 1 node has failed. One option is to choose the best $k$ nodes to repair. This serves as a base-line without opportunistic repair. With opportunistic repair, an optimal $d$ number of nodes are chosen. We see the  average ratio of  time to repair with opportunistic repair and the  time to repair from $k$ nodes in Figure \ref{fig:time_repair_onefail} (which is given as $\mathbb{E}\frac{[k^{th} {\text largest bandwidth}]}{[\max_{k\le d<n} (d-k+1) \times d^{th} \text{largest bandwidth}]}$. We assume  that first node fails, and the average is taken  over the different choices of bandwidths), where the average is over 1000 runs for the bandwidths between different nodes. For a fixed $k=5$, Figure \ref{fig:time_repair_onefail} gives this mean ratio for different values of $n$ (Figure \ref{fig:time_repair_onefail_std} gives the standard deviation for the ratio). As $n$ increases, the options for using greater than $k$ nodes increase so that the relative time taken to repair with opportunistic repair is smaller.
\begin{figure}[htb] 
   \centering
   \includegraphics[trim =3in 1.5in 2.8in 1.2in, clip,angle=90,width=3in]{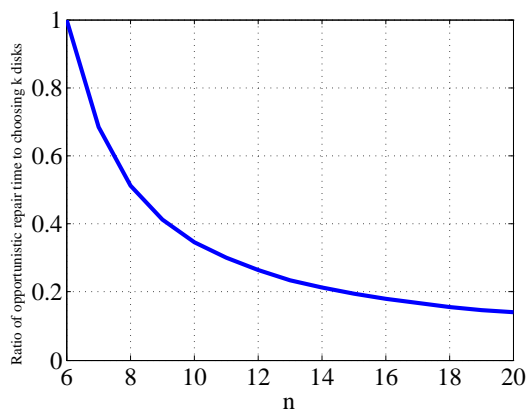}
   \caption{Average ratio of time taken to repair 1 failed node with opportunistic repair to the repair time using $k$ nodes.}
   \label{fig:time_repair_onefail}
\end{figure}

\begin{figure}[htb] 
   \centering
   \includegraphics[trim =3in 1.5in 2.8in 1.1in, clip,angle=90,width=3in]{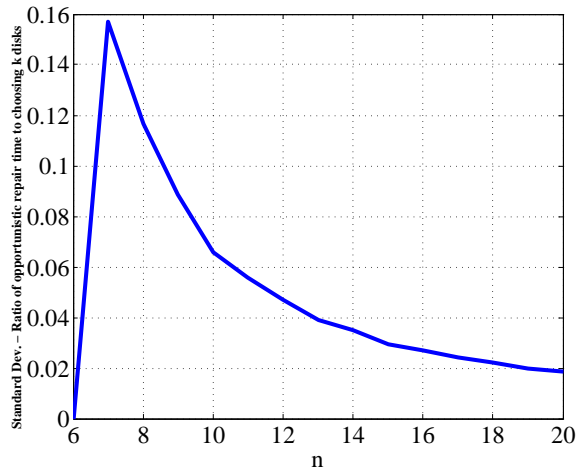}
   \caption{Standard deviation of the ratio of time taken to repair 1 failed node with opportunistic repair to the repair time using $k$ nodes.}
   \label{fig:time_repair_onefail_std}
\end{figure}

\begin{figure}[htb] 
   \centering
   \includegraphics[trim =3.4in 1.5in 3.5in 1.5in, clip,angle=90,width=3in]{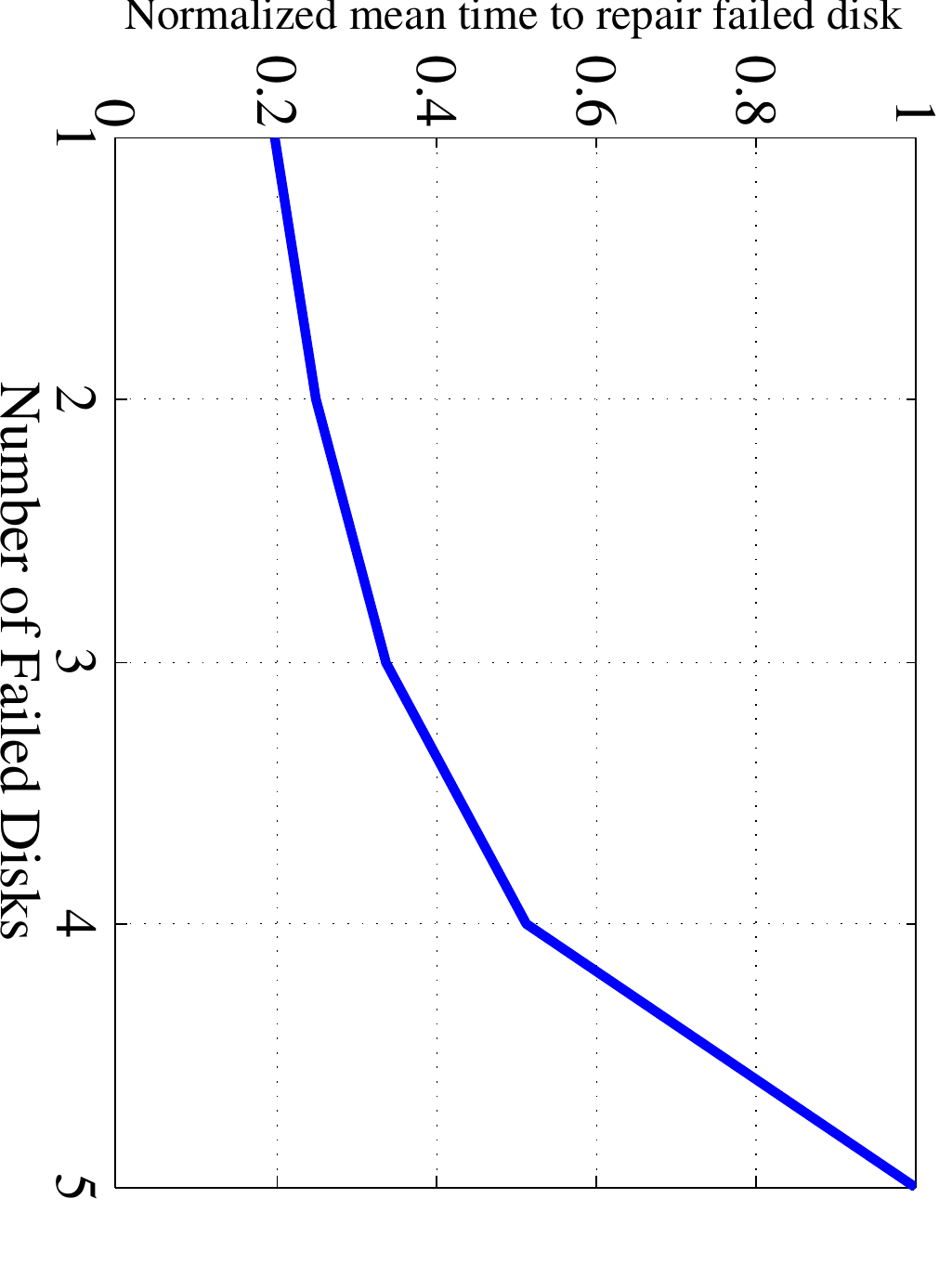}
   \caption{Average of normalized time taken to repair a failed node with opportunistic repair when multiple nodes have failed. The normalization is taken such that time taken to repair $n-k$ failed nodes is unit.}
   \label{fig:time_repair_multifail}
\end{figure}

\begin{figure}[htb] 
   \centering
   \includegraphics[trim =3.4in 1.5in 3.5in 1.4in, clip,angle=90,width=3in]{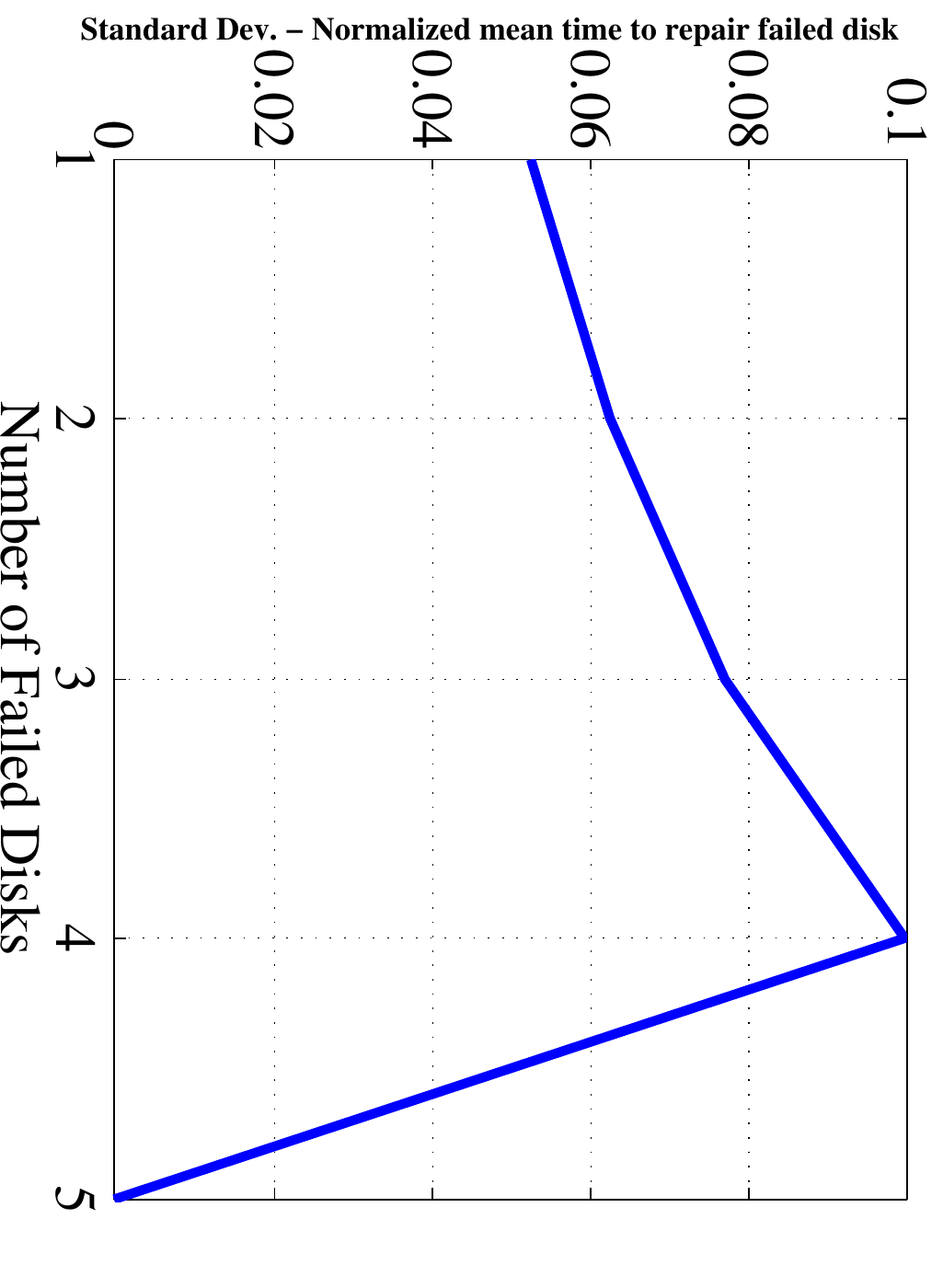}
   \caption{Standard deviation of normalized time taken to repair a failed node with opportunistic repair when multiple nodes have failed. The normalization is taken such that time taken to repair $n-k$ failed nodes is unit.}
   \label{fig:time_repair_multifail_std}
\end{figure}

We next consider $k=5$ and $n=10$. For this system,  the repair time when $1, \cdots, 4$ nodes fail is compared to the repair time when $5$ nodes fail. If $t$ nodes fail, the repair time is $\frac{c}{\max_{k\le d\le n-t} (d-k+1) B_d} $, where $B_d$ is  $d^{th}$ largest bandwidth among remaining $n-t$ nodes to the node to be repaired, for a constant $c>0$. Thus, if more nodes fail, the maximization in the denominator is over smaller range and thus the time taken to repair with opportunistic repair is larger. This is depicted in Figure \ref{fig:time_repair_multifail}, where we consider $t$ nodes failing from $1$ to $t$, and the average mean time to repair first node when $t$ nodes fail (where the average is over $10^6$ different bandwidth link configurations) is depicted. The average mean time to repair when $n-k=5$ nodes fail is normalized to unit (Figure \ref{fig:time_repair_multifail_std} gives the standard deviation for the normalized repair time).  This saving in bandwidth when less nodes fail illustrate significant savings in the mean time to data loss of the system, since the factor of improvement is intuitively the product of savings for each $d$. The calculation of exact mean time to data loss would involve calculations of failure and repair for different nodes, and since the bandwidths to a disk are not symmetric, this calculation is involved and thus not considered in the paper.



\section{Conclusions}

This paper describes a  distributed erasure coded storage system with the capability that a failed node can be repaired from a number of helper nodes $d$ that is not fixed a priori,
and investigates the repair bandwidth vs. storage tradeoff for such a system.
 This paper then demonstrate the usefulness of opportunistic repair in the form of an improvement in the mean time to data loss of the system and show that the improvement is significant even when different nodes have random bandwidth links.

In this paper, we only consider functional repair for opportunistic distributed storage systems. Even though exact repair codes has been shown to exist asymptotically at the MSR point \cite{Cadambe:11}, general constructions for exact repair are still open. We have assumed a mesh network to present the benefits of opportunistic repair, the benefits from other network topologies is an open problem. In a general network, finding the constraints on the bandwidth region between different node pairs via a probing technique is needed to be able to decide the number of nodes to access in order to repair the failed node.


\appendix
\section{Proof of  Theorem 1} \label{apdx_inner}

The proof follows a similar line as that in \cite{Dimakis:10}. As in \cite{Dimakis:10}, we construct an information flow graph which is a directed acyclic graph,
consisting of three kinds of nodes: a single data source $S$,
storage nodes $x_{in}^i$, $x_{out}^i$  and data collectors $DC_i$. The single
node $S$ corresponds to the source of the original data. Storage
node $i$ in the system is represented by a storage input node
$x_{in}^i$, and a storage output node $x_{out}^i$; these two nodes are
connected by a directed edge $x_{in}^i \to x_{out}^i$ with capacity equal
to the amount of data stored at node $i$.

Given the dynamic nature of the storage systems that we
consider, the information ﬂow graph also evolves in time. At
any given time, each vertex in the graph is either active or
inactive, depending on whether it is available in the network.
At the initial time, only the source node $S$ is active; it then
contacts an initial set of storage nodes, and connects to their
inputs ($x_{in}$) with directed edges of inﬁnite capacity. From
this point onwards, the original source node $S$ becomes and
remains inactive. At the next time step, the initially chosen
storage nodes become now active; they represent a distributed
erasure code, corresponding to the desired state of the system.
If a new node $j$ joins the system, it can only be connected
with active nodes. If the newcomer $j$ chooses to connect
with active storage node $i$, then we add a directed edge from
$x_{out}^i$ to $x_{in}^j$, with capacity equal to the amount of information
communicated from node $i$ to the newcomer. Finally, a data collector $DC$ is a node that corresponds
to a request to reconstruct the data. Data collectors connect to
subsets of active nodes through edges with infinite capacity.

An important notion associated with the information flow
graph is that of minimum cuts: A (directed) cut in the graph G
between the source $S$ and a fixed data collector node $DC$ is a
subset $C$ of edges such that, there is no directed path starting
from $S$ to $DC$ that does not have one or more edges in $C$. The
minimum cut is the cut between $S$ and $DC$ in which the total
sum of the edge capacities is smallest.

Following the approach of \cite{Dimakis:10}, it is enough to prove the following Lemma.
\begin{lemma}
Consider any (potentially infinite) information flow graph $G$, formed by having $n$ initial nodes that connect
directly to the source and obtain $\alpha$ bits, while additional nodes
join the graph by connecting to $d_j\in D$ existing nodes and obtaining
$\beta_{d_j}$ bits from each for some $d_j\in D$. Any data collector $t$ that connects to a $k$-subset of ``out-nodes'' of G must satisfy:
\begin{equation}
\text{mincut}(s,t) \ge \sum_{i=0}^{k-1}\min(\alpha, \min_{d_j \in D}(d_j -i)\beta_{d_j}).
\end{equation}
Furthermore, there exists an information flow graph $G^*$
 where this bound is matched with equality.
\end{lemma}

Let $e_i = \arg\min_{d_j \in D}(d_j -i)\beta_{d_j}$ for $i=0, 1, \cdots k-1$. For the statement that there exist a flow where the bound is matched with equality, we consider the setup as in Figure \ref{mincut}. In this graph, there are initially $n$ nodes labeled from $1$ to $n$. Consider $k$ newcomers labeled as
$n+1, \cdots, n+k$. The newcomer node $n+i$ connects to nodes
$n+i-e_{i-1}$, $\cdots$, $n+i-1$. Consider a data collector $t$ that connects
to the last $k$ nodes, i.e., nodes $n+1, \cdots, n+k$, and a
cut $(U,\bar{U})$ defined as follows. For each $i \in  \{1, \cdots , k\}$, if
$\alpha \le (e_i-1)\beta_{e_i}$, then we include $x_{out}^{n+i}$ in $\bar{U}$; otherwise, we
include $x_{out}^{n+i}$ and $x_{in}^{n+i}$ in $U$.  
We note that this cut $(U,\bar{U})$ achieves
the bound as in the statement of Lemma with equality.
\begin{figure*}[tb]
  \centering
      \includegraphics[trim = .15in 1in .25in 1in, clip, width=0.75\textwidth]{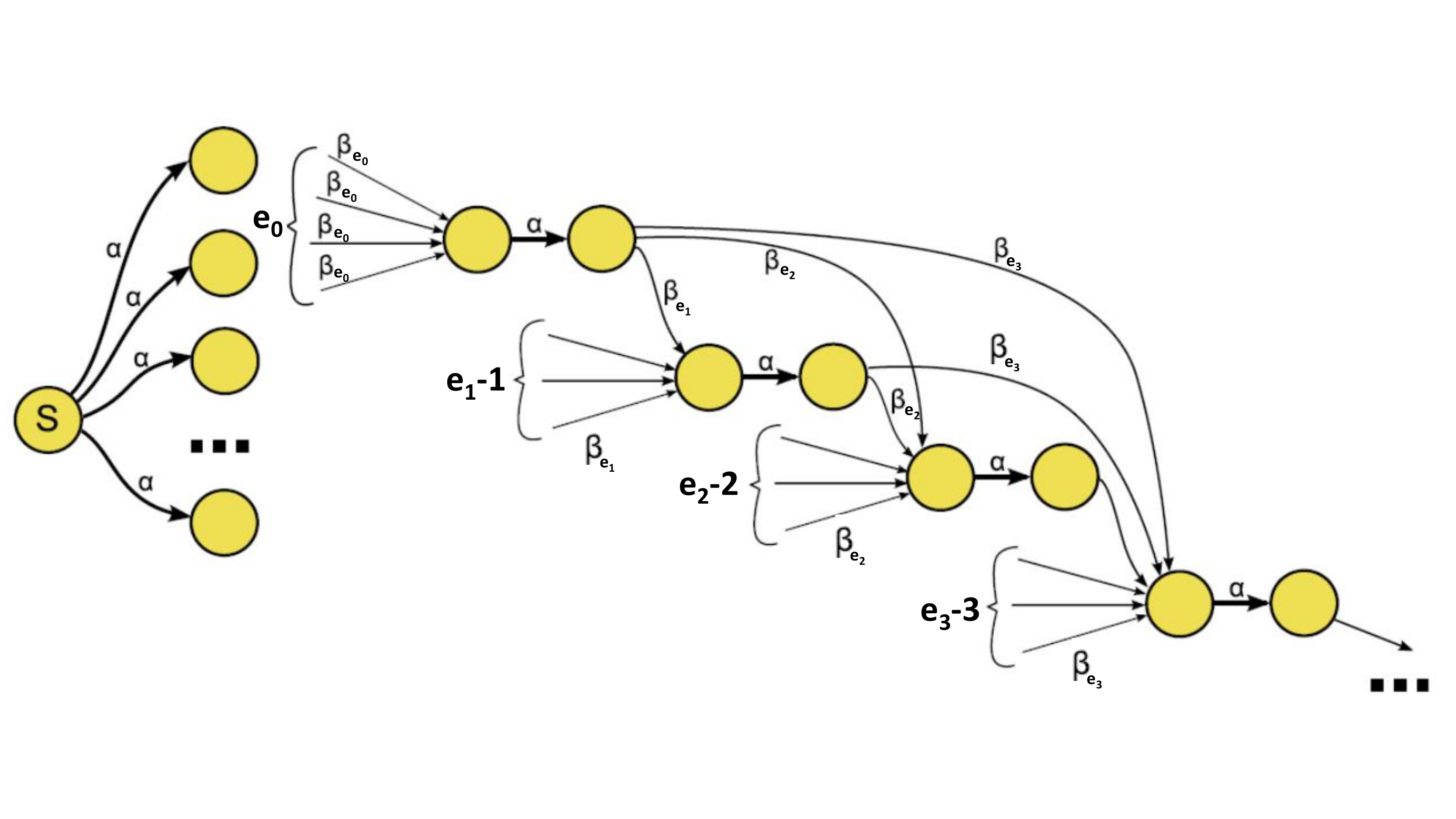}
  \caption{$G^*$ used in the proof of lemma.}\label{mincut}
\end{figure*}

The proof that every cut should satisfy the bound follows very similarly to the proof in \cite{Dimakis:10}, using the topological sorting for the graph, and is thus omitted.

\section{Proof of Theorem \ref{theorem:noloss}}
\begin{proof}
To prove this result, we start with a subset of $D$, say $D'= \{e_1, e_2\}$ for any $\{e_1,e_2\}\subseteq D$ with $e_1 > e_2$. Since $(\alpha, \beta_{e_1}^*(\alpha),  \infty)$ satisfy \eqref{main}, there is a minimum  ${\beta_{e_2}}(\alpha)$ such that $(\alpha, \beta_{e_1}^*(\alpha), {\beta_{e_2}}(\alpha) )$ satisfy \eqref{main}. We call this minimum ${\beta_{e_2}}(\alpha)$ as  $\widetilde{\beta_{e_2}}(\alpha)$. For ${\beta_{e_2}}(\alpha)=\widetilde{\beta_{e_2}}(\alpha)$, \eqref{main} will be satisfied with equality since if not, the value of $\widetilde{\beta_{e_2}}(\alpha)$ is not optimal.

Thus, we have the following equations
\begin{eqnarray}
\sum_{i=0}^{k-1}\min(\alpha, (e_1 -i)\beta_{e_1}^*(\alpha)) &=& \mathsf{M},\\
\sum_{i=0}^{k-1}\min(\alpha, (e_1 -i)\beta_{e_1}^*(\alpha), (e_2 -i)\widetilde{\beta_{e_2}}(\alpha)) &=& \mathsf{M},
\end{eqnarray}

Since each term inside the summation in the second expression is at-most that in the first expression, we have
\begin{equation}
(e_2 -i)\widetilde{\beta_{e_2}}(\alpha) \ge \min(\alpha, (e_1 -i)\beta_{e_1}^*(\alpha)),
\end{equation}
for all $0\le i\le k-1$. Since, $\widetilde{\beta_{e_2}}(\alpha)$ is the minimum possible ${\beta_{e_2}}(\alpha)$ satisfying the above, we have
\begin{eqnarray}
\widetilde{\beta_{e_2}}(\alpha) &=& \min_{i=0}^{k-1} \frac{\min(\alpha, (e_1 -i)\beta_{e_1}^*(\alpha))}{e_2-i}\\
&=&  \min_{i=0}^{k-1} \min(\frac{\alpha}{e_2-i}, \frac{e_1 -i}{e_2-i}\beta_{e_1}^*(\alpha))
\end{eqnarray}
Since both the terms in the minimum increase with $i$, we have that the minimum of these terms is non-decreasing with $i$, and thus
\begin{equation}
\widetilde{\beta_{e_2}}(\alpha) = \min(\frac{\alpha}{e_2-k+1}, \frac{e_1 -k+1}{e_2-k+1}\beta_{e_1}^*(\alpha))
\end{equation}
Further, we note that $\alpha < (e_1 -k+1)\beta_{e_1}^*(\alpha)$ is not possible since it violates the optimality of $\beta_{e_1}^*(\alpha)$ and thus, we have
\begin{equation}
\widetilde{\beta_{e_2}}(\alpha) = \frac{e_1 -k+1}{e_2-k+1}\beta_{e_1}^*(\alpha)
\end{equation}

It now remains to be seen as to when is $\widetilde{\beta_{e_2}}(\alpha) = \beta_{e_2}^*(\alpha)$. If $\widetilde{\beta_{e_2}}(\alpha) = \beta_{e_2}^*(\alpha)$, we have the following
\begin{equation}
\sum_{i=0}^{k-1}\min(\alpha, (e_2 -i)\frac{e_1 -k+1}{e_2-k+1}\beta_{e_1}^*(\alpha)) = \mathsf{M}.
\end{equation}
Since we know that $\sum_{i=0}^{k-1}\min(\alpha, (e_1 -i)\beta_{e_1}^*(\alpha)) = \mathsf{M}$ and $(e_2 -i)\frac{e_1 -k+1}{e_2-k+1}\ge (e_1 -i)$, we have that
\begin{equation}
\min(\alpha, (e_2 -i)\frac{e_1 -k+1}{e_2-k+1}\beta_{e_1}^*(\alpha)) = \min(\alpha, (e_1 -i)\beta_{e_1}^*(\alpha)),
\end{equation}
for all $0\le i \le k-1$. Since for $i=k-1$, the two sides are exactly the same and thus the above holds for $k=1$.  Thus for $k>1$, we need the two sides to be equal for $0\le i\le k-2$.  Since for $0\le i<k-1$, $e_1 -i < (e_2 -i)\frac{e_1 -k+1}{e_2-k+1}$, we have that the above holds if and only if
\begin{equation}
\alpha \le (e_1-k+2)\beta_{e_1}^*(\alpha).
\end{equation}

From the expression of $\beta_{e_1}^*(\alpha)$, we have that this happens if and only if $\alpha \le \frac{\mathsf{M}(e_1-k+2)}{k(e_1-k+2)-1}$.

Thus, we see that if $\alpha > \frac{\mathsf{M}(d_1-k+2)}{k(d_1-k+2)-1}$, $(\alpha, \beta_{d_1}^*(\alpha),  \cdots, \beta_{d_l}^*(\alpha))$ do not satisfy \eqref{main}. Further, if $\alpha \le \frac{\mathsf{M}(d_1-k+2)}{k(d_1-k+2)-1}$, the same approach shows that $(\alpha, \beta_{d_1}^*(\alpha),  \cdots, \beta_{d_l}^*(\alpha))$ satisfies \eqref{main}.
\end{proof}

\section{Proof of  Theorem \ref{thm_MTTDL}} \label{apdx_MTTDL}

Let $P_c(t)$ denote the probability that $c$ nodes are active at time $t$. The differential equations corresponding to the change of state are given by
\begin{eqnarray}
\frac{dP_n(t)}{dt} &=& -n\lambda P_n(t) + (n-k)\mu P_{n-1}(t) \\
\frac{dP_{n-c}(t)}{dt} &=& (n-c+1)\lambda P_{n-c+1}(t) -((n-c)\lambda \nonumber\\&&
+ (n-k-c+1)\mu) P_{n-c}(t)\nonumber\\&& + (n-k-c)\mu P_{n-c-1}(t)\nonumber\\&&(\text{ for $1\le c \le n-k+1$})\\
\frac{dP_{k}(t)}{dt} &=& (k+1)\lambda P_{k+1}(t) - (k\lambda + \mu) P_k(t) \\
\frac{dP_{k-1}(t)}{dt} &=& k\lambda P_k(t).
\end{eqnarray}

Let $\overline{P}(t) = \left[\begin{matrix}
  P_n(t)  \\
  \cdots  \\
  P_{k-1}(t)
 \end{matrix}\right]$.

Solving the differential equations, we have $\overline{P}(t) = \exp(At)\overline{P}(0)$, where $A$ is a tri-diagonal matrix with each column sum as zero, and is given as $A=$
\begin{equation}
\left[\begin{matrix}
  -n\lambda & (n-k)\mu & 0 & \cdots & 0 & 0 \\
n\lambda & -(n-1)\lambda-(n-k)\mu & (n-k)\mu & \cdots & \cdots & 0 \\
\cdots & \cdots & \cdots & \cdots & 0 & 0 \\
0 & 0 & 0 & \cdots & -k\lambda-\mu & 0 \\
0 & 0 & 0 & \cdots & k\lambda & 0
 \end{matrix}\right].
\end{equation}

Since the system starts from the state with all nodes being active, we have $\overline{P}(0)$ being 1 only in the first element, and zero elsewhere. The mean time to data loss is given as
\begin{equation}
MTTDL = \int_{t=0}^\infty \sum_{r = k}^n P_r(t) dt
\end{equation}

Note that by Final Value Theorem for Laplace Transform \cite{laplace}, we have
\begin{equation}
MTTDL = \lim_{s\to 0} [1 1 1 \cdots 1 0] (sI-A)^{-1}[1 0 \cdots 0]^T.
\end{equation}

We let $B \triangleq adj(sI-A)$, where $adj(.)$ represents adjoint of the argument, and let $B_{ij}$ be the element corresponding to $i^{th}$ row and $j^{th}$ column. Then,

\begin{equation}
MTTDL = \lim_{s\to 0} \frac{\sum_{r=1}^{n-k+1}B_{1r}}{\det(sI-A)} \label{main_MTTDL}.
\end{equation}

We will now give a result that is a key result in evaluation of the determinant and the adjoints.

\begin{lemma}\label{matrix_det}
Let $M$ be a $l \times l $ tri-diagonal matrix, such that the $r^{th}$ diagonal element is $c_r\lambda + b_r \mu$, the upper diagonal element $M_{r,r+1} = -b_{r+1}\mu$, and the lower diagonal element $M_{r+1,r} = -c_r \lambda$. Then, the determinant of matrix $M$ is given as
\begin{equation}
\det(M) = \sum_{i=0}^l (b_1 b_2\cdots b_{l-i})(c_{l-i+1}\cdots c_l)\mu^{l-i}\lambda^i
\end{equation}
\end{lemma}
\begin{proof}
The result can be shown to hold by induction and thus the proof is omitted.
\end{proof}

When we construct matrix $sI-A$, we see that there is no element in the last row or last column except the diagonal element which is $s$. Thus, the determinant of $sI-A$ is given as $s$ times the determinant of the first $n-k+1 \times n-k+1$ matrix. Further, this is equal to $s$ times the determinant of the first $n-k+1 \times n-k+1$ matrix when $s=0$ $+ o(s)$. For $s=0$, the first $n-k+1 \times n-k+1$ is given in the tri-diagonal form in Lemma \ref{matrix_det} with $b_1=0$, $c_r = (n-r+1)$. Thus, we have
\begin{eqnarray}
\det(sI-A) &=& s n(n-1)\cdots k \lambda^{n-k+1} + o(s) \nonumber\\&=& s \frac{n!}{(k-1)!} \lambda^{n-k+1} + o(s)
\end{eqnarray}

Using similar approach, we see that
\begin{eqnarray}
B_{11} &=& s \sum_{i=0}^{n-k}\left((n-k)(n-k-1)\cdots \right.\nonumber\\&&\left. (n-k-i+1)\right)\nonumber\\&&\left((n-i-1)\cdots k \right)\mu^i \lambda^{n-k-i} +o(s)\\
&=&s \frac{(n-k)!}{(k-1)!}\sum_{i=0}^{n-k}\mu^i \lambda^{n-k-i} \nonumber\\&&\frac{(n-i-1)!}{(n-k-i)!} + o(s)
\end{eqnarray}

Let $F_n \triangleq \frac{(n-k)!}{(k-1)!}\sum_{i=0}^{n-k}\mu^i \lambda^{n-k-i} \frac{(n-i-1)!}{(n-k-i)!}$. Then, $B_{11} = sF_n +o(s)$.
Similarly solving other terms, we have
$B_{1(l+1)}= s \frac{n!}{(n-l)!} \lambda^l F_{n-l} +o(s)$ for $l=0, \cdots n-k$. Thus, the overall mean time to data loss is given as
\begin{eqnarray}
MTTDL &=& \lim_{s\to 0} \frac{\sum_{r=1}^{n-k+2}B_{1r}}{\det(sI-A)} \\
&=& \lim_{s\to 0} \frac{\sum_{l=0}^{n-k}B_{1(l+1)}}{s \frac{n!}{(k-1)!} \lambda^{n-k+1} + o(s)} \\
&=& \lim_{s\to 0} \frac{s \sum_{l=0}^{n-k} \frac{n!}{(n-l)!} \lambda^l F_{n-l} + o(s)}{s \frac{n!}{(k-1)!} \lambda^{n-k+1} + o(s)} \\
&=& \frac{\sum_{l=0}^{n-k} \frac{n!}{(n-l)!} \lambda^l F_{n-l}}{ \frac{n!}{(k-1)!} \lambda^{n-k+1} }
\end{eqnarray}
Solving this expression gives the result as in the statement of the Theorem after some manipulations.



\end{document}